\documentclass[a4paper,USenglish]{lipics-v2019}
\usepackage[utf8]{inputenc}
\usepackage{nccmath}

\newcommand{\N}{\mathbb{N}}
\newcommand{\Z}{\mathbb{Z}}
\newcommand{\bool}{\ensuremath{\mathsf{bool}}}
\newcommand{\encode}{\ensuremath{\mathsf{encode}}}
\newcommand{\decode}{\ensuremath{\mathsf{decode}}}
\newcommand{\option}{\ensuremath{\mathsf{option}}}
\newcommand{\some}{\ensuremath{\mathsf{some}}}
\newcommand{\None}{\ensuremath{\mathsf{none}}}
\newcommand{\Part}{\ensuremath{\mathsf{part}}}
\newcommand{\ls}{\ensuremath{\mathsf{list}}}
\newcommand{\inl}{\ensuremath{\mathsf{inl}}}
\newcommand{\inr}{\ensuremath{\mathsf{inr}}}
\newcommand{\ifz}{\ensuremath{\mathsf{ifz}}}
\newcommand{\eval}{\ensuremath{\mathsf{eval}}}
\newcommand{\curry}{\ensuremath{\mathsf{curry}}}
\newcommand{\encodable}{\ensuremath{\mathsf{encodable}}}
\newcommand{\primcodable}{\ensuremath{\mathsf{primcodable}}}
\newcommand{\code}{\ensuremath{\mathsf{code}}}
\newcommand\pto{\mathrel{\ooalign{\hfil$\mapstochar\mkern5mu$\hfil\cr$\to$\cr}}}

\usepackage{color}
\definecolor{keywordcolor}{rgb}{0.7, 0.1, 0.1}   
\definecolor{commentcolor}{rgb}{0.4, 0.4, 0.4}   
\definecolor{symbolcolor}{rgb}{0.0, 0.1, 0.6}    
\definecolor{sortcolor}{rgb}{0.1, 0.5, 0.1}      

\usepackage{listings}

\lstMakeShortInline"

\title{Formalizing computability theory via partial recursive functions}

\titlerunning{Formalizing computability theory}

\author{Mario Carneiro}{Carnegie Mellon University, Pittsburgh, PA, USA}{mcarneir@andrew.cmu.edu}{https://orcid.org/0000-0002-0470-5249}{}

\authorrunning{M. Carneiro}

\Copyright{Mario Carneiro}


\category{}

\relatedversion{A full version of the paper is available at \url{https://arxiv.org/abs/1810.08380}.}

\supplement{\label{sec:supplemental}The formalization is a part of the \texttt{mathlib} Lean library at\\ \url{https://github.com/leanprover-community/mathlib}, and a snapshot as of this publication is available at \url{http://github.com/digama0/mathlib-ITP2019}.}

\funding{This material is based upon work supported by AFOSR grant FA9550-18-1-0120 and a grant from the Sloan Foundation.}

\acknowledgements{I would like to thank my advisor Jeremy Avigad for his support and encouragement, and for his reviews of early drafts of this work, as well as Yannick Forster, Rob Y. Lewis,  and the anonymous reviewers for their many helpful comments.}

\nolinenumbers 

\hideLIPIcs  

\EventEditors{John Harrison, John O'Leary, and Andrew Tolmach}
\EventNoEds{3}
\EventLongTitle{10th International Conference on Interactive Theorem Proving (ITP 2019)}
\EventShortTitle{ITP 2019}
\EventAcronym{ITP}
\EventYear{2019}
\EventDate{September 9--12, 2019}
\EventLocation{Portland, OR, USA}
\EventLogo{}
\SeriesVolume{141}
\ArticleNo{19}

\begin{document}

\maketitle

\begin{abstract}
We present an extension to the \texttt{mathlib} library of the Lean theorem prover formalizing the foundations of computability theory. We use primitive recursive functions and partial recursive functions as the main objects of study, and we use a constructive encoding of partial functions such that they are executable when the programs in question provably halt. Main theorems include the construction of a universal partial recursive function and a proof of the undecidability of the halting problem. Type class inference provides a transparent way to supply G\"{o}del numberings where needed and encapsulate the encoding details.
\end{abstract}

\section{Introduction}

Computability theory is the study of the limitations of computers, first brought into focus in the 1930s by Alan Turing by his discoveries on the existence of universal Turing machines and the unsolvability of the halting problem \cite{turing1937}, and Alonso Church with the $\lambda$-calculus as a model of computation \cite{church1936}. Together with Kleene's $\mu$-recursive functions \cite{kleene1943}, that these all give the same collection of ``computable functions'' gave credence to the thesis \cite{church1936} that this is the ``right'' notion of computation, and that all others are equivalent in power. Today, this work lies at the basis of programming language semantics and the mathematical analysis of computers.

Like many areas of mathematics, computability theory remains somewhat ``formally ambiguous'' about its foundations, in the sense that most theorems and proofs can be stated with respect to a number of different concretizations of the ideas in play. This can be considered a feature of informal mathematics, because it allows us to focus on the essential aspects without getting caught up in details which are more an artifact of the encoding than aspects that are relevant to the theory itself, but it is one of the harder things to deal with as a formalizer, because definitions must be made relative to \emph{some} encoding, and this colors the rest of the development.

In computability theory, we have three or four competing formulations of ``computable,'' which are all equivalent, but each present their own view on the concept. As a pragmatic matter, Turing machines have become the de facto standard formulation of computable functions, but they are also notorious for requiring a lot of tedious encoding in order to get the theory off the ground, to the extent that the term ``Turing tarpit'' is now used for languages in which ``everything is possible but nothing of interest is easy.'' \cite{perlis1982} Asperti and Riccioti \cite{asperti2012} have formalized the construction of a universal Turing machine in Matita, but the encoding details make the process long and arduous. Norrish \cite{norrish2011} uses the lambda calculus in HOL4, which is cleaner but still requires some complications with respect to the handling of partiality and type dependence.

Instead, we build our theory on Kleene's theory of $\mu$-recursive functions. In this theory, we have a collection of functions $\N^k\to \N$, in which we can perform basic operations on $\N$, as well as recursive constructions on the natural number arguments. This produces the primitive recursive functions, and adding an unbounded recursion operator $\mu x. P(x)$ gives these functions the same expressive power as Turing-computable functions. We hope to show that the ``main result'' here, the existence of a universal machine, is easiest to achieve over the partial recursive functions, avoiding the complications of explicit substitution in the $\lambda$-calculus and encoding tricks in Turing Machines, and moreover that the usage of typeclasses for G\"{o}del numbering provides a rich and flexible language for discussing computability over arbitrary types.

This theory has been developed in the Lean theorem prover, a relatively young proof assistant based on dependent type theory with inductive types, written primarily by Leonardo de Moura at Microsoft Research \cite{demoura2015}. The full development is available in the \texttt{mathlib} standard library (see the \hyperref[sec:supplemental]{Supplemental Material}). In Section \ref{sec:encodable} we describe our extensible approach to G\"{o}del numbering, in Section \ref{sec:primrec} we look at primitive recursive functions, extended to partial recursive functions in Section \ref{sec:partrec}. Section \ref{sec:universal} deals with the universal partial recursive function and its properties, including its application to unsolvability of the halting problem.

\section{Encodable sets}\label{sec:encodable}
As mentioned in the introduction, we would like to support some level of formal ambiguity when encoding problems, such as defining languages as subsets of $\N$ vs. subsets of $\{0,1\}^*$, or even $\Sigma^*$ where $\Sigma$ is some finite or countable alphabet. Similarly, we would like to talk about primitive recursive functions of type $\Z\times\Z\to\Z$, or the partial recursive function $\eval:\code\times\N\pto\N$ that evaluates a partial function specified by a code (see Section \ref{sec:universal}).

Unfortunately it is not enough just to know that these types are countable. While the exact bijection to $\N$ is not so important, it is important that we not use one bijection in a proof and a different bijection in the next proof, because these differ by an automorphism of $\N$ which may not be computable. (For example, if we encode the halting Turing machines as even numbers and the non-halting ones as odd numbers, and then the halting problem becomes trivial.) In complexity theory it becomes even more important that these bijections are ``simple'' and do not smuggle in any additional computational power.

To support these uses, we make use of Lean's typeclass resolution mechanism, which is a way of inferring structure on types in a syntax-directed way. The major advantage of this approach is that it allows us to fix a uniform encoding that we can then apply to all types constructed from a few basic building blocks, which avoids the multiple encoding problem, and still lets us use the types we would like to (or even construct new types like $\mathsf{code}$ whose explicit structure reflects the inductive construction of partial recursive functions, rather than the encoding details).

\begin{figure}
\begin{center}
\begin{tabular}{c|ccccc}
 & 0 & 1 & 2 & 3 & \dots \\\hline
0 & 0 & 1 & 4 & 9 &\\
1 & 2 & 3 & 5 & 10 &\\
2 & 6 & 7 & 8 & 11 &\\
3 & 12 & 13 & 14 & 15& \\
$\vdots$ & & & & & $\ddots$
\end{tabular}
\end{center}
\caption{The pairing function \lstinline"mkpair a b = if a < b then b*b + a else a*a + a + b".}
\label{fig:mkpair}
\end{figure}

At the core of this is the function $\mathsf{mkpair}:\N\times\N\to\N$, and its inverse $\mathsf{unpair}:\N\to\N\times\N$ forming a bijection (see Figure \ref{fig:mkpair}). There is very little we need about these functions except their definability, and that $\mathsf{mkpair}$ and the two components of $\mathsf{unpair}$ are primitive recursive.

\begin{figure}
\centering
\begin{lstlisting}
class encodable (α : Type u) :=
(encode : α → nat)
(decode : nat → option α)
(encodek : ∀ a, decode (encode a) = some a)

variables {α β} [encodable α] [encodable β] 
def encode_sum : α ⊕ β → ℕ
| (inl a) := 2 * encode a
| (inr b) := 2 * encode b + 1

def encode_prod : α × β → ℕ
| (a, b) := mkpair (encode a) (encode b)

def encode_option : option α → ℕ
| none     := 0
| (some a) := succ (encode a)

\end{lstlisting}
\caption{The \encodable\ typeclass, and some example definitions of the encoding functions for the disjoint sum and product operators on types. (The corresponding $\decode$ functions are omitted.)}
\label{fig:encodable}
\end{figure}

We say that a type $\alpha$ is \emph{encodable} if we have a function $\encode:\alpha\to\N$, and a partial inverse $\decode:\N\to\option\;\alpha$ which correctly decodes any value in the image of $\encode$. Here $\option\;\alpha$ is the type consisting of the elements $\some\;a$ for $a:\alpha$, and an extra element \None\ representing failure or undefinedness. If the $\decode$ function happens to be total (that is, never returns \None), then $\alpha$ is called \emph{denumerable}. Importantly, these notions are ``data'' in the sense that they impose additional structure on the type -- there are nonequivalent ways for a type to be \encodable, and we will want these properties to be inferred in a consistent way. (This definition does not originate with us; Lean has had the \encodable\ typeclass almost since the beginning, and MathComp has a similar class, called \textsf{Countable}.)

Classically, an \encodable\ instance on $\alpha$ is just an injection to $\N$, and a \textsf{denumerable} instance is just a bijection to $\N$. But constructively these are not equivalent, and since these notions lie in the executable fragment of Lean (they don't use any classical axioms), one can actually run these encoding functions on concrete values of the types, i.e. we can evaluate $\encode\;(\some\;(2, 3)) = 12$.

\section{Primitive recursive functions}\label{sec:primrec}
\begin{figure}
\begin{lstlisting}
inductive primrec : (ℕ → ℕ) → Prop
| zero : primrec (λ n, 0)
| succ : primrec succ
| left : primrec (λ n, fst (unpair n))
| right : primrec (λ n, snd (unpair n))
| pair {f g} : primrec f → primrec g →
  primrec (λ n, mkpair (f n) (g n))
| comp {f g} : primrec f → primrec g →
  primrec (f ∘ g)
| prec {f g} : primrec f → primrec g →
  primrec (unpaired (λ z n, nat.rec_on n (f z)
     (λ y IH, g (mkpair z (mkpair y IH)))))
\end{lstlisting}
\caption{The definition of primitive recursive on $\N$ in Lean. The \textsf{unpaired} function turns a function $\N\to\N\to\N$ into $\N\to\N$ by composing with \textsf{unpair}, and $\mathsf{nat.rec\_on}:\N\to\alpha\to(\N\to \alpha\to\alpha)\to\alpha$ is Lean's built-in recursor for $\N$.}
\label{fig:primrec}
\end{figure}
The traditional definition of primitive recursive functions looks something like this:
\begin{definition}\label{def:primrec}
The primitive recursive functions are the least subset of functions $\N^k\to\N$  satisfying the following conditions:
\begin{itemize}
\item The function $n\mapsto 0$ is prim.\ rec.
\item The function $n\mapsto n+1$ is prim.\ rec.
\item The function $(n_0,\dots,n_{k-1})\mapsto n_i$ is prim.\ rec. for each $0\le i<k$.
\item If $f:\N^k\to\N$ and $g_i:\N^m\to\N$ for $i\le k$ are prim.\ rec., then so is the $n$-way composition $v\mapsto f(g_0(v),\dots,g_{k-1}(v))$.
\item If $f:\N^m\to\N$ and $g:\N^{m+2}\to\N$ are prim.\ rec., then the function $h:\N^{m+1}\to\N$ defined by
\begin{ceqn}
\begin{align*}
h(\vec z,0)&=f(\vec z)\\
h(\vec z,n+1)&=g(\vec z,n,h(\vec z, n))
\end{align*}
\end{ceqn}
is also prim.\ rec.
\end{itemize}
\end{definition}
CIC is quite good at expressing these kinds of constructions as inductively defined predicates. See Figure \ref{fig:primrec} for the definition that appears in Lean. But there is an important difference in this formulation: rather than dealing with $n$-ary functions, we utilize the pairing function on $\N$ to write everything as a function $\N\to\N$ with only one argument. This drastically simplifies the composition rule to just the usual function composition, and in the primitive recursion rule we need only one auxiliary parameter $z:\N$ rather than $\vec z:\N^m$. Then the projection functions are replaced with the \textsf{left} and \textsf{right} cases for the components of $\mathsf{unpair}:\N\to\N\times\N$, and in order to express composition with higher arity functions, we need the \textsf{pair} constructor to explicitly form the map $x\mapsto(f\;x,g\;x)$. (See Section \ref{sec:textbook} if you think this definition is a cheat.)

Now that we have a definition of `primitive recursive' that works for functions on $\N$, we would like to extend it to other types using the \encodable\ mechanism discussed in Section \ref{sec:encodable}. There is a problem though, because given an arbitrary \encodable\ instance we can combine the $\decode:\N\to\option\;\alpha$ with the function $\encode:\option\;\alpha\to\N$ defined on $\option\;\alpha$ induced by this \encodable\ instance to form a new function $\encode\circ\decode:\N\to\N$, which may or may not be primitive recursive. If it is not, then it brings new power to the primitive recursive functions and so it is not a pure translation of \textsf{primrec} to other types. To resolve this, we define $\mathsf{primcodable}\;\alpha$ to mean exactly that $\alpha$ has an \encodable\ instance for which this composition is primitive recursive. All of the \encodable\ constructions we have discussed (indeed, all those defined in Lean) are \primcodable, so this is not a severe restriction.

Now we can say that a function between arbitrary \textsf{primcodable} types is primitive recursive if when we pass $f$ through the $\encode$ and $\decode$ functions we get a primitive recursive function on $\N$:
\begin{lstlisting}
def primrec {α β} [primcodable α] [primcodable β] (f : α → β) : Prop :=
nat.primrec (λ n, encode (option.map f (decode α n)))
\end{lstlisting}
\emph{Note:} The function $\mathsf{option.map}$ lifts $f$ to a function on option types before applying it to $\decode$. The result has type $\option\;\beta$, which has an \encode\ function because $\beta$ does.

Now we are in a position to recover the textbook definition of primitive recursive, because $\N^k$ is \primcodable, so we have the language to say that $f:\N^k\to\N$ is primitive recursive, and indeed this is equivalent to Definition \ref{def:primrec}.

But we can now say much more: the $\some:\alpha\to\option\;\alpha$ function is primitive recursive because it is just encoded as $\mathsf{succ}$. The constant function $\lambda a. b:\alpha\to\beta$ is primitive recursive because it encodes to some constant function (composed with a function that filters out values not in the domain $\alpha$). The composition of prim.\ rec.\ functions on arbitrary types is prim.\ rec. The pair of primitive recursive functions $\lambda a.\,(f\;a,g\;a)$, where $f:\alpha\to\beta$ and $g:\alpha\to\gamma$, is primitive recursive.

Indeed all the usual basic operations on inductive types like \textsf{sum}, \textsf{prod}, and \textsf{option} are primitive recursive. We define convenient syntax $\mathsf{primrec}_2$ for prim.\ rec.\ binary functions $\alpha\to\beta\to\gamma$ (a common case), expressed by uncurrying to $\alpha\times\beta\to\gamma$, and $\mathsf{primrec\_pred}$ for primitive recursive predicates $\alpha\to\mathsf{Prop}$, which are decidable predicates which are primitive recursive when coerced to $\bool$ (which is \encodable).

The big caveat comes in theorems like the following:
\begin{quote}
If $\alpha$ and $\beta$ are \primcodable\ types and $f:\alpha\to\beta$ and $g:\alpha\to\N\to\beta\to\beta$ are prim.\ rec., then the function $h:\alpha\to\N\to\beta$ defined by
\begin{ceqn}
\begin{align*}
h\;a\;0&=f\;a\\
h\;a\;(n+1)&=g\;a\;n\;(h\;a\;n)
\end{align*}
\end{ceqn} is also prim.\ rec.
\end{quote}
This is of course just the generalization of the primitive recursion clause to arbitrary types, but it requires that the target type be \primcodable, which means in particular that it is countable, so we cannot define an object of function type by recursion. (The universal partial recursive function will give us a way to get around this later.) But this is in some sense ``working as intended,'' since this is exactly why the Ackermann function
\begin{ceqn}
\begin{align*}
A(0,n) &= n+1\\
A(m+1,0) &= A(m,1)\\
A(m+1,n+1) &= A(m,A(m+1,n))
\end{align*}
\end{ceqn}
is not primitive recursive.

Another restriction placed on us relative to Lean's built-in notion of primitive recursion on $\N$ is that that while $\mathsf{nat.rec\_on}$ has a dependent type, we have no mechanism for supporting dependent types via \encodable. We follow the tradition of HOL based provers here and encode dependencies using $\option$ types so we can fail on a garbage input. However, it is possible to support a dependent family via a separate typeclass. For example we could define $\primcodable_2\;F$, where $F:\alpha\to\mathsf{Type}$ and $\alpha$ is \encodable, to mean that $\Pi a, \encodable\;(F\;a)$, and moreover this family of $\encode$/$\decode$ functions is prim.\ rec.\ jointly in both arguments. In the end we did not pursue this because of the added complexity and lack of compelling use cases.

One other \primcodable\ type we have not yet discussed is $\ls\;\alpha$, the type of finite lists of values of type $\alpha$. The $\encode$ and $\decode$ functions are defined recursively via the bijection $\ls\;\alpha\simeq\option\;(\alpha\times\ls\;\alpha)$. (Note that this is not a particularly good encoding for complexity theory, as it grows super-exponentially in the length of the list.) Even without using this instance, we can prove that any function $f:\alpha\to\beta$ is prim.\ rec.\ when $\alpha$ is finite, by getting the elements of $\alpha$ as a list, and writing $f$ as the composition of an index lookup of $a_i$ in $[a_0,\dots,a_{n-1}]$ and the $i$th element function in $[f\;a_0,\dots,f\;a_{n-1}]$ to map $a_i$ to $f\;a_i$.

The proof that $\primcodable\;(\ls\;\alpha)$ is a bit delicate. The definition of the encode/decode functions in Lean is a well-founded recursion, but to show it is primitive recursive we must construct the function without any higher-order features. First, we prove that the $\mathsf{foldl}:(\alpha\to\beta\to\alpha)\to\alpha\to\ls\;\beta\to\alpha$ function is prim.\ rec.\ when its arguments are. To do this, given $f:\alpha\to\beta\to\alpha$, we construct an accumulator $\alpha\times\ls\;\beta$ with the initial inputs, and then repeatedly transform it so that $(a,[])\mapsto(a,[])$ and $(a,b::l)\mapsto(f\;a\;b,l)$. Since the encoding scheme satisfies $\encode\;l\ge \mathsf{length}\;l$ for all lists $l$, if we iterate this map $\encode\;l$ times, we exhaust the input list and the accumulator will contain the desired result. We can then use \textsf{foldl} to define \textsf{reverse}, and combine them to define \textsf{foldr}, which is what we need to define the primcodable function for $\ls\;\alpha$.

Complicating matters, we needed a primcodable instance for $\primcodable\;(\ls\;\alpha)$ to state the original theorem that $\mathsf{foldl}$ is prim.rec., so we have a circularity. To resolve this, we use $\ls\;\N$ as a bootstrap, which is trivially primcodable because it is denumerable.

Once we allow the list itself to be an input, we get some more interesting possibilities. In particular, the function $\mathsf{list.nth}:\ls\;\alpha\to\N\to\option\;\alpha$, which gets an element from a list by index (or returns \None\ if the index is out of bounds), is primitive recursive, and this fact expresses an equivalent of G\"{o}del's sequence number theorem \cite{godel1931} (for a different encoding than G\"{o}del's original encoding). From this we can prove the following ``strong recursion'' theorem:
\begin{lstlisting}
theorem nat_strong_rec
  (f : α → ℕ → σ)
  {g : α → list σ → option σ}
  (hg : primrec₂ g)
  (H : ∀ a n, g a (map (f a) (range n)) = some (f a n)) :
  primrec₂ f
\end{lstlisting}
Ignoring the parameter $a$, the main hypothesis says essentially that $f(n)=g(f\restriction[0,\dots,n-1])$, where the first $n$ values of $f$ have been written in a list (and the length of the list tells $g$ what value of $f$ we are constructing). The reason $g$ has optional return value is to allow for it to fail when the input is not valid.

Once we have lists, the dependent type $\mathsf{vector}\;\alpha\;n$ is just a subtype of $\ls\;\alpha$, so it has an easy \primcodable\ instance, and most of the vector functions follow from their list counterparts. Similarly for functions $\mathsf{fin}\;n\to\alpha$, which are isomorphic to $\mathsf{vector}\;\alpha\;n$.

\subsection{The textbook definition}\label{sec:textbook}
Now that we have a proper theory, we can return to the question of how to show equivalence to Definition  \ref{def:primrec}. We do this by defining $\mathsf{nat.primrec}':\forall n, (\mathsf{vector}\;\N\;n\to\N)\to\mathsf{Prop}$ with only 5 clauses matching Definition \ref{def:primrec}. It is easy to show at this point that $\mathsf{primrec'}$ implies $\mathsf{primrec}$, since all of the functions appearing in Definition \ref{def:primrec} are known to be primitive recursive. For the converse, most of the clauses are easy, but our earlier cheat was to axiomatize that \textsf{mkpair} and \textsf{unpair} are primitive recursive, even though the definition involves addition, multiplication and case analysis in \textsf{mkpair} and even square root in the inverse function (see Figure \ref{fig:unpair}).
\begin{figure}[t]
\begin{lstlisting}
def unpair (n : ℕ) : ℕ × ℕ :=
let s := sqrt n in
if n - s*s < s then (n - s*s, s) else (s, n - s*s - s)
\end{lstlisting}
\caption{The function $\mathsf{unpair} : \N\to\N\times\N$. (Here $\mathsf{sqrt}:\N\to\N$ is actually the function $n\mapsto\lfloor \sqrt n\rfloor$.)}
\label{fig:unpair}
\end{figure}
So we must show that all these operations are primitive recursive by the textbook definition. The square root case is not as difficult as it may sound; since it grows by at most 1 at each step we can define it by primitive recursion as
\begin{ceqn}
\begin{align*}
\lfloor \sqrt 0\rfloor&= 0\\
\lfloor \sqrt{n+1}\rfloor&=\mbox{if }n+1<(y+1)^2\mbox{ then }y\mbox{ else }y+1\\
&\qquad\mbox{where }y=\lfloor \sqrt n\rfloor.
\end{align*}
\end{ceqn}

This alternate basis for \textsf{primrec} is useful for reductions, for example, to show that some other basis for computation like Turing machines can simulate every primitive recursive function.

\section{Partial recursive functions}\label{sec:partrec}
The partial recursive functions are an extension of primitive recursive functions by adding an operator $\mu n.\;p(n)$, where $p:\N\to\mathsf{bool}$ is a predicate, which denotes the least value of $n$ such that $p(n)$ is true. Intuitively, this value is found by starting at 0 and testing ever larger values until a satisfying instance is found. This function is not always defined, in the sense that even when all the inputs are well typed it may not return a value -- it can result in an ``infinite loop.''

Before we tackle the partial recursive functions we must understand partiality itself, and in particular how to represent unbounded computation, computably, in a proof assistant that can only represent terminating computations. As Lean is based on dependent type theory, which is strongly normalizing, all expression evaluation terminates, and so the problem is prima facie unsolvable -- we may as well turn to functional relations as a representation. However, as we shall see, it is actually possible with no additional modifications to CIC or extra axioms.

\subsection{The partiality monad}
We have already discussed the $\option\;\alpha$ type for representing a possible failure state, but nontermination is a slightly different kind of ``failure'' in that the program is not able to tell that it has failed while executing, and this difference makes itself known in the type system.

To address this distinction, we introduce the $\Part\;\alpha$ type:
\begin{lstlisting}
def part (α : Type*) := ∑ p : Prop, (p → α)
\end{lstlisting}
That is, an element $p:\Part\;\alpha$ is a dependent pair of a proposition $p_1$ and a function $p_2:p_1\to \alpha$ from proofs of $p_1$ to $\alpha$. A value of type $\Part\;\alpha$ is a nondecidable optional value, in the sense that there is not necessarily a decision procedure for determining if the $\Part\;\alpha$ contains a value, but if it does then you can extract the value using the function component. This type has a monad structure, as follows:
\begin{ceqn}
\begin{align*}
&\mathsf{pure}:\alpha\to\Part\;\alpha\\
&\mathsf{pure}\;a=\langle \mathsf{true},\lambda\_.\;a\rangle\\
&\mathsf{bind}:\Part\;\alpha\to(\alpha\to\Part\;\beta)\to\Part\;\beta\\
&\mathsf{bind}\;\langle p,f\rangle\;g=\langle(\exists h:p,(g\;(f\;h))_1),(\lambda h.\;(g\;(f\;h_1))_2\;h_2)\rangle
\end{align*}
\end{ceqn}
Also, there is an element $\bot=\langle\mathsf{false},\mathsf{exfalso}\rangle:\Part\;\alpha$ representing an undefined value. We can map $\option\;\alpha\to\Part\;\alpha$ by sending $\some\;a$ to $\mathsf{pure}\;a$ and $\None$ to $\bot$, and assuming the law of excluded middle in Type we can also define an inverse map and show $\option\;\alpha\simeq\Part\;\alpha$, but this breaks the computational interpretation of $\Part\;\alpha$.

The definition of \textsf{bind}, also written in Haskell style as the infix operator \texttt{>>=}, is slightly intricate but is ``exactly what you would expect'' in terms of its behavior. Given a partial value $p:\Part\;\alpha$ and a function $f:\alpha\to\Part\;\beta$, the resulting partial value $p\mbox{ \texttt{>>=} }f:\Part\;\beta$ is defined when $p$ is defined to be some $a:\alpha$, and $f\;a$ is defined, in which case it evaluates to $f\;a$.

It is convenient to abstract from the definition to a relational version, where $a\in p$ means $\exists h:p_1,p_2\;h=a$ -- that is, $a\in p$ says that $p$ is defined and equal to $a$. (This relation is functional because of proof irrelevance.) With this definition the bind operator can be much more easily expressed by the theorem
$$b\in p\mathrel{\texttt{>>=}}f\leftrightarrow\exists a\in p,b\in f\;a$$
which is shared with many other collection-based monad structures. Also, like every other monad there is a \textsf{map} operator, written \texttt{<\$>}, which applies a pure function to a partial value:
\begin{ceqn}
\begin{align*}
&\mathsf{map}:(\alpha\to\beta)\to\Part\;\alpha\to \Part\;\beta\\
&f\mathrel{\mbox{\texttt{<\$>}}}p=\langle p_1,f\circ p_2\rangle
\end{align*}
\end{ceqn}
Because they come up often, we will use the notation $\alpha\pto\beta=\alpha\to\Part\;\beta$ for the type of all partial functions from $\alpha$ to $\beta$.

One important function that is (constructively) definable on this type is $\textsf{fix}$, which has the following properties:
\begin{ceqn}
\begin{align*}
&\mathsf{fix}\;(f:\alpha\pto\beta\oplus\alpha):\alpha\pto\beta\\
&b\in \mathsf{fix}\;f\;a\leftrightarrow\inl\;b\in f\;a\lor\exists a',\inr\;a'\in f\;a\land b\in\mathsf{fix}\;f\;a'
\end{align*}
\end{ceqn}
Given an input $a$, it evaluates $f$ to get either $\inl\;b$ or $\inr\;a'$. In the first case it returns $b$, and in the second case it starts over with the value $a'$. The function $\mathsf{fix}\;f$ is defined when this process eventually terminates with a value, if we assume this then we can construct the value that $\mathsf{fix}\;f$ returns. So even though Lean's type theory does not permit unbounded recursion, by working in this partiality monad we get computable unbounded recursion.

The minimization operator $\mathsf{find}\;p=\mu n.\,p(n)$, which finds the smallest value satisfying the (partial) boolean predicate $p$ can be defined in terms of \textsf{fix} as follows:
\begin{ceqn}
\begin{align*}
&\mathsf{find}:(\N\pto\bool)\pto\N\\
&\mathsf{find}\;p=\mathsf{fix}\;(\lambda n.\;\mbox{if }p\;n\mbox{ then }\mathsf{inl}\;n\mbox{ else }\mathsf{inr}(n+1))\;0
\end{align*}
\end{ceqn}

As an aside, we note that while this monad supports many of the operations one expects on partial recursive functions, one thing it does not support is parallel computation. That is, we would like to have a nondeterministic choice function $\mbox{\texttt{<|>}}:\Part\;\alpha\to\Part\;\alpha\to\Part\;\alpha$ such that $p\mathrel{\mbox{\texttt{<|>}}}q$ is defined if either $p$ or $q$ is defined (with value arbitrarily chosen from the two). This is possible for partial recursive functions, but it is not constructively definable for $\Part$. For this, we must restrict the propositions to be \emph{semidecidable} \cite{bauer2006}, which means essentially that they are a $\Sigma_1$ proposition, that is, a proposition of the form $\exists n.\;f(n)=\mathsf{true}$ for some $f:\N\to\mathsf{bool}$. Every partial recursive function is semidecidable as a consequence of the $\mathsf{eval}_k$ function (see Section \ref{sec:evaln}).

\subsection{\textsf{partrec} and \textsf{computable}}
\begin{figure}
\begin{lstlisting}
inductive partrec : (ℕ ⇸ ℕ) → Prop
| zero : partrec (pure 0)
| succ : partrec succ
| left : partrec (λ n, fst (unpair n))
| right : partrec (λ n, snd (unpair n))
| pair {f g} : partrec f → partrec g →
  partrec (λ n, f n >>= λ a, g n >>= λ b, pure (mkpair a b))
| comp {f g} : partrec f → partrec g →
  partrec (λ n, g n >>= f)
| prec {f g} : partrec f → partrec g →
  partrec (unpaired (λ a n, nat.rec_on n (f a)
    (λ y IH, IH >>= λ i,
      g (mkpair a (mkpair y i)))))
| find {f} : partrec f → partrec (λ a,
  find (λ n, (λ m, m = 0) <$> f (mkpair a n)))
\end{lstlisting}
\caption{The definition of partial recursive on $\N$ in Lean.}
\label{fig:partrec}
\end{figure}
The definition \textsf{nat.partrec} is given in Figure \ref{fig:partrec}. The first 7 cases are almost the same as those of \textsf{primrec}, except that we must now worry about partiality in all the operations that build functions. So for example \lstinline"λ n, f n >>= λ a, g n >>= λ b, pure (mkpair a b)" is the function $n\mapsto(f\;n,g\;n)$ except that if the computation of either $f\;n$ or $g\;n$ fails to return a value, then this is not defined. (In other words, this operation is ``strict'' in both arguments). Similarly, the composition is now expressed as \lstinline"λ n, g n >>= f", which says that $g\;n$ should be evaluated first, and if it is defined and equals $a$, then $f\;a$ is the resulting value.

The interesting case is the last one, which incorporates the \textsf{find} function on $\N$. Ignoring partiality, it says that $\lambda a.\,\mu n.\,f(a,n)=0$ is partial recursive if $f$ is. This is of course the source of the partiality -- all the other constructors produce total functions from total functions but this can be partial if the function $f$ is never zero.

Although this defines a class of partial functions, some of the functions happen to be total anyway, and we call a total partial-recursive function \emph{computable}. It is an easy fact that every primitive recursive function is computable.

As before, we can compose with $\encode$ and $\decode$ to extend these definitions to any \primcodable\ type. Although we could define an analogue of \primcodable\ using computable functions instead of primitive recursive functions, since we want to stick to simple encodings (usually not just primitive recursive but polynomial time), and we already have encodings for all the important types, so \primcodable\ is enough.

One aspect of this definition which is not obviously a problem until one works out all the details is the strictness of the \textsf{prec} constructor. In conventional notation, it says that if $f:\alpha\pto\beta$ and $g:\alpha\to\N\to\beta\pto\beta$ are partial recursive functions, then so is the function $h:\alpha\to\N\pto\beta$ defined by
\begin{ceqn}
\begin{align*}
h(a,0)&=f(a)\\
h(a,n+1)&=g(a,n,h(a,n)).
\end{align*}
\end{ceqn}
Importantly, $h(a,n+1)$ is only defined if $h(a,n)$ is defined and $g(a,n,h(a,n))$ is defined. It does not matter if $g$ does not make use of the argument at all, for example if it is the first projection. This comes up in the definition of the lazy conditional $\ifz[f,g]$, defined when $f:\alpha\pto\beta$, $g:\alpha\pto\beta$ by:
\begin{ceqn}
\begin{align*}
&\ifz[f,g]:\alpha\to\N\pto\beta\\
&\ifz[f,g](a,n)=\begin{cases}
f(a)&\mbox{if }n=0\\
g(a)&\mbox{if }n\ne 0
\end{cases},
\end{align*}
\end{ceqn}
where in particular $\ifz[f,g](a,1)=g(a)$ regardless of whether $f(a)$ is defined. This is the basis of ``if statements'' that resemble execution paths in a computer -- we need a way to choose which subcomputation to perform, without needing to evaluate both. The usual way of implementing $\ifz$ is to use primitive recursion on the argument $n$, using $f$ in the zero case and $g\circ\pi_1$ in the successor case. But because of the strictness constraint, this will result in $\ifz[\bot,g](a,1)=(g\circ\pi_1)(a,0,f(a))=\bot$ because $f(a)=\bot$, rather than the desired result $g(a)$. In fact, we won't have the tools to solve this problem until Section \ref{sec:apps}.

\section{Universality}\label{sec:universal}
\subsection{Codes for functions}
Because \textsf{partrec} is an inductive predicate, we can read off a corresponding data type of syntactic representations witnessing that a function $\N\pto\N$ is partial recursive:
\begin{lstlisting}
inductive code : Type
| zero : code
| succ : code
| left : code
| right : code
| pair : code → code → code
| comp : code → code → code
| prec : code → code → code
| find' : code → code
\end{lstlisting}
We can define the semantics of a code via an ``evaluation'' function that takes a \textsf{code} and an input value in $\N$ and produces a partial $\N$ value.
\begin{lstlisting}
def eval : code → ℕ ⇸ ℕ
| zero         := pure 0
| succ         := succ
| left         := λ n, n.unpair.1
| right        := λ n, n.unpair.2
| (pair cf cg) := λ n,
  eval cf n >>= λ a, eval cg n >>= λ b, pure (mkpair a b)
| (comp cf cg) := λ n, eval cg n >>= eval cf
| (prec cf cg) := unpaired (λ a n,
  nat.rec_on n (eval cf a) (λ y IH, IH >>= λ i,
    eval cg (mkpair a (mkpair y i))))
| (find' cf)  := unpaired (λ a m, (λ i, i + m) <$>
  find (λ n, (λ m, m = 0) <$> eval cf (mkpair a (n + m))))
\end{lstlisting}
Then it is a simple consequence of the definition that $f$ is partial recursive iff there exists a code $\hat f$ such that $f=\eval\;\hat f$.

\emph{Note:} The $\mathsf{find'}$ constructor is a slightly modified version of $\mathsf{find}$ which is easier to use in evaluation:
$$\mathsf{find'}\;f\;(a,m)=(\mu n.\;f(a,n+m)=0)+m,$$
which can be expressed in terms of $\mathsf{find}$ as:
\begin{ceqn}
\begin{align*}
\mathsf{find}\;f\;a&=\mathsf{find'}\;f\;(a,0)\\
\mathsf{find'}\;f\;(a,m)&=\mathsf{find}\;(\lambda x.\,f\,(x_1,x_2+m))\;a+m
\end{align*}
\end{ceqn}
So we can pretend that \textsf{partrec} was defined with a case for $\mathsf{find'}$ instead of $\mathsf{find}$ since it yields the same class of functions.

Now the key fact is that \code\ is \textsf{denumerable}. Concretely, we can encode it using a combination of the tricks we used to encode sums, products and option types, that is,
\begin{ceqn}
\begin{align*}
\encode\,(\mathsf{zero})&=0\\
\encode\,(\mathsf{succ})&=1\\
\encode\,(\mathsf{left})&=2\\
\encode\,(\mathsf{right})&=3\\
\encode\,(\mathsf{pair}\;c_1\;c_2)&=4\cdot(\encode\;c_1,\encode\;c_2)+4\\
\encode\,(\mathsf{comp}\;c_1\;c_2)&=4\cdot(\encode\;c_1,\encode\;c_2)+5\\
\encode\,(\mathsf{prec}\;c_1\;c_2)&=4\cdot(\encode\;c_1,\encode\;c_2)+6\\
\encode\,(\mathsf{find'}\;c)&=4\cdot(\encode\;c)+7
\end{align*}
\end{ceqn}
where $(m,n)$ is the pairing function from Figure \ref{fig:mkpair}. (We could have used a more permissive encoding, but this has the advantage that it is a bijection to $\N$, which makes the proof that this is a \primcodable\ type trivial.)

Having shown that the type is \primcodable\ we can now start to show that functions \emph{on codes} are primitive recursive. In particular, all the constructors are primitive recursive, the recursion principle preserves primitive recursiveness and computability (not partial recursiveness, because of the as-yet unresolved problem with \ifz), and we can prove that these simple functions on codes are primitive recursive:
\begin{ceqn}
\begin{align*}
&\mathsf{const}:\N\to\code\\
&\eval\;(\mathsf{const}\;a)\;n=a\\
&\curry:\code\to\N\to\code\\
&\eval\;(\curry\;c\;m)\;n=\eval\;c\;(m,n)
\end{align*}
\end{ceqn}
In particular, the rather understated fact that $\curry$ is primitive recursive is a form of the $s$-$m$-$n$ theorem of recursion theory.

\subsection{Resource-bounded evaluation}\label{sec:evaln}
We have one more component before the universality theorem. We define a ``resource-bounded'' version of $\eval$, namely $\eval_k:\code\to\N\to\option\;\N$ where $k:\N$. (In the formal text it is called \texttt{evaln}.) This function is total -- we have a definite failure condition this time, unlike $\eval$ itself, which can diverge. There are multiple ways to define this function; the important part is that if $\eval\;c\;n=\bot$ then $\eval_k\;c\;n=\None$ for all $k$, and if $\eval\;c\;n=a$ is defined then $\eval_k\;c\;n=\some\;a$ for some $k$. Furthermore, it is convenient to ensure that $\eval_k$ is monotonic in $k$, and the domain of $\eval_k$ is contained in $[0,k]$, that is, if $n>k$ then $\eval_k\;c\;n=\None$.

\begin{figure}[tb]
\begin{lstlisting}
def evaln : ∀ k : ℕ, code → ℕ → option ℕ
| 0 _ := λ n, none
| (k+1) zero := λ n, guard (n ≤ k) >> pure 0
| (k+1) succ := λ n, guard (n ≤ k) >> pure (succ n)
| (k+1) left := λ n, guard (n ≤ k) >> pure (fst (unpair n))
| (k+1) right := λ n, guard (n ≤ k) >> pure (snd (unpair n))
| (k+1) (pair cf cg) := λ n, guard (n ≤ k) >>
  evaln (k+1) cf n >>= λ a, evaln (k+1) cg n >>= λ b, pure (mkpair a b)
| (k+1) (comp cf cg) := λ n, guard (n ≤ k) >>
  evaln (k+1) cg n >>= λ x, evaln (k+1) cf x
| (k+1) (prec cf cg) := λ n, guard (n ≤ k) >>
  unpaired (λ a m, nat.rec_on m
    (evaln (k+1) cf a) 
    (λ y, evaln k (prec cf cg) (mkpair a y) >>= λ i,
          evaln (k+1) cg (mkpair a (mkpair y i)))) n
| (k+1) (find' cf)  := λ n, guard (n ≤ k) >>
  unpaired (λ a m,
    evaln (k+1) cf (mkpair a m) >>= λ x,
    if x = 0 then pure m else
      evaln k (find' cf) (mkpair a (m+1))) n
\end{lstlisting}
\caption{The definition of resource-bounded evaluation of partial recursive functions in Lean. \emph{Notation note:} The \texttt{>>} operator is monad sequencing, i.e. $a\mathrel{\texttt{>>}}b=a\mathrel{\texttt{>>=}}\lambda\_.\;b$, and $\mathtt{guard}\ p:\mathsf{option}\;\mathsf{unit}$ is the function that returns $\mathsf{some}\;()$ if $p$ is true and $\mathsf{none}$ if $p$ is false. Together they ensure that \texttt{evaln k c n = none} unless $n\le k$.}
\label{fig:evaln}
\end{figure}

The Lean definition of \texttt{evaln} is given in Figure \ref{fig:evaln}. The details of the definition are not so important, but it is interesting to note that our ``fuel'' $k$ for the computation only needs to decrease when we don't change the program code in the recursive call, namely in the \textsf{prec} and \textsf{find'} cases, thanks to Lean's pattern matcher (which compiles this definition into one by nested structural recursion). (You may wonder why we cannot use the fact that $n$ is decreasing in the \textsf{prec} case to prove termination, but this is because the function is not defined by recursion on $n$, it is by recursion on $k$ at all $n\le k$ simultaneously.)

Because $\eval_k\;c:\N\to\option\;\N$ has finite domain $n\in [0,k]$ outside which it is \None, we can encode the whole function as a single $\ls\;(\option\;\N)$. Thus we can pack the function into the type $\N\times\code\to\ls\;(\option\;\N)$, and define this by strong recursion (using the theorem \texttt{nat\_strong\_rec} mentioned in Section \ref{sec:primrec}), since in every case of the recursion, either $k$ decreases and $c$ remains fixed, or $c$ decreases and $k$ remains fixed.

Thus $\mathsf{evaln}:\N\to\code\to\N\to\option\;\N$ is primitive recursive (jointly in all arguments), and since $\eval\;c\;n=\eval_{k'}\;c\;n$ where $k'=\mu k.\;(\eval_{k'}\;c\;n\ne\None)$, this shows that $\eval$ is partial recursive. This is Kleene's normal form theorem (in a different language) -- $\eval$ is a universal partial recursive function.

\subsection{Applications}\label{sec:apps}
The fixed point theorems are an easy consequence of universality. These have all been formalized; the formalized theorem names are given in parentheses.
\begin{theorem}[\texttt{fixed\_point}]\label{thm:fixed_point}
If $f:\code\to\code$ is computable, then there exists some code $c$ such that $\eval(f\;c)=\eval\;c$.
\end{theorem}
\begin{proof}
Consider the function $g:\N\to\N\pto\N$ defined by $g\;x\;y=\eval\;(\eval\;x\;x)\;y$ (using $\decode:\N\to\code$ to use natural numbers as codes in $\eval$). This function is clearly partial recursive, so let $g=\eval\;\hat g$. Now let $F:\N\to\code$ such that $F\;x=f\;(\curry\;\hat g\;x)$; then $F$ is computable so let $F=\eval\;\hat F$. Then for $c=\curry\;\hat g\;\hat F$ we have:
\begin{ceqn}
\begin{align*}
\eval\;(f\;c)\;n&=\eval\;(f\;(\curry\;\hat g\;\hat F))\;n\\
&=\eval\;(F\;\hat F)\;n\\
&=\eval\;(\eval\;\hat F\;\hat F)\;n\\
&=g\;\hat F\;n\\
&=\eval\;\hat g\;(\hat F,n)\\
&=\eval\;(\curry\;\hat g\;\hat F)\;n\\
&=\eval\;c\;n.\\[-40pt]
\end{align*}
\end{ceqn}
\end{proof}

\begin{theorem}[$\mathtt{fixed\_point_2}$]\label{thm:fixed_point2}
If $f:\code\to\N\pto\N$ is partial recursive, then there exists some code $c$ such that $\eval\;c=f\;c$.
\end{theorem}
\begin{proof}
Let $f=\eval\;\hat f$, and apply Theorem \ref{thm:fixed_point} to $\curry\;\hat f$ to obtain a $c$ such that\\ $\eval\;(\curry\;\hat f\;c)=\eval\;c$. Then
\begin{ceqn}
\begin{align*}
\eval\;c\;n&=\eval\;(\curry\;\hat f\;c)\;n\\
&=\eval\;\hat f\;(c,n)\\
&=f\;c\;n.\\[-40pt]
\end{align*}
\end{ceqn}
\end{proof}
We can also finally solve the $\ifz$ problem. If $f$ and $g$ are partial recursive functions, then letting $f=\eval\;\hat f$ and $g=\eval\;\hat g$, the function
$$c(n)=\begin{cases}
\hat f&\mbox{if }n=0\\
\hat g&\mbox{if }n\ne 0
\end{cases}$$
is primitive recursive (since both branches are just numbers now instead of computations that may not halt), and $\ifz[f,g](a,n)=\eval\;c(n)\;a$. More generally, this implies that we can evaluate conditionals where the condition is a computable function and the branches are partial functions. We can also construct a nondeterministic choice function:

\begin{theorem}[\texttt{merge}]\label{thm:merge}
If $f,g:\alpha\pto\beta$ are partial recursive functions, then there exists a function $h:\N\pto\N$ such that $h(a)$ is defined iff either $f(a)$ or $g(a)$ is defined, and if $x\in h(a)$ then $x\in f(a)$ or $x\in g(a)$.
\end{theorem}
\begin{proof}
It is easy to reduce to the case where $f,g:\N\pto\N$. Let $f=\eval\;\hat{f}$ and $g=\eval\;\hat{g}$; then $h(n)=\mathsf{find}(\lambda k.\;\eval_k\;\hat{f}\;n\mathrel{\mbox{\texttt{<|>}}}\eval_k\;\hat{g}\;n)$ works, where \mbox{\texttt{<|>}} is the alternative operator on $\mathsf{option}\;\N$.
\end{proof}

A corollary is Post's theorem on the equivalence of computable and r.e. co-r.e. sets:
\begin{theorem}[\texttt{computable\_iff\_re\_compl\_re}]\label{thm:post}
If $p:\alpha\to\mathsf{Prop}$ is a decidable predicate, then $p$ is computable iff $p$ is r.e. and $\lambda a.\;\neg p\;a$ is r.e.
\end{theorem}
\begin{proof}
The forward direction is trivial. In the reverse direction, if $f,g:\alpha\pto\mathsf{unit}$ are chosen such that $f(a)$ is defined iff $p(a)$ and $g(a)$ is defined iff $\neg p(a)$, then by Theorem \ref{thm:merge} there is a function $h:\alpha\pto\mathsf{bool}$ extending $\lambda a.\;f(a)\mathrel{\texttt{>>}}\mathsf{pure}\;\mathsf{true}$ and $\lambda a.\;g(a)\mathrel{\texttt{>>}}\mathsf{pure}\;\mathsf{false}$. This function has domain $\{a\mid p(a)\lor\neg p(a)\}=\alpha$ (because $p$ is decidable) and is $\mathsf{true}$ when $p(a)$ is true and is $\mathsf{false}$ when $\neg p(a)$. Thus $h$ is a computable indicator function for $p$.
\end{proof}

The assumption that $p$ is decidable is not the tightest condition we could assert; it suffices $p$ is stable, i.e. $\neg\neg p(a)\to p(a)$, or alternatively we could assume Markov's principle or LEM.

We conclude with Rice's theorem on the noncomputability of all nontrivial properties about computable functions:
\begin{theorem}[\texttt{rice}]\label{thm:rice}
Let $C\subseteq (\N\pto\N)$ such that $\{c\mid\eval\;c\in C\}$ is computable. Then for any $f,g:\N\pto\N$, $f\in C$ implies $g\in C$ (so classically $C=\emptyset\lor C=\N\pto\N$).
\end{theorem}
\begin{proof}
Apply Theorem \ref{thm:fixed_point2} to the function $F\;c\;n=\mbox{if }\eval\;c\in C\mbox{ then }g\;n\mbox{ else }f\;n.$ to obtain a $c$ such that $\eval\;c=F\;c$. (Note $\eval\;c\in C$ is decidable because it is computable.) Then if $\eval\;c\in C$, we have $F\;c\;n=g\;n$ for all $n$ so $\eval\;c=F\;c=g$, hence $g\in C$. And if $\eval\;c\notin C$ then $\eval\;c=F\;c=f$ similarly which contradicts $f\in C$, $\eval\;c\notin C$.
\end{proof}

The undecidability of the halting problem is a trivial corollary:
\begin{theorem}[\texttt{halting\_problem}]\label{thm:halting}
The set $\{c:\code\mid \eval\;c\;0\mbox{ is defined}\}$ is not computable.
\end{theorem}
\begin{proof}
Suppose it is; we can write it as $\{c\mid \eval\;c\in C\}$ where $C=\{f\mid f\;0\mbox{ is defined}\}$, so applying Rice's theorem with $f=\lambda n.\;0$ and $g=\lambda n.\;\bot$ we have a contradiction from $f\in C$ and $g\notin C$.
\end{proof}

\begin{figure}[t]
\centering
\begin{tabular}{r|l|r}
    File & Section & Line Count \\\hline
    \texttt{primrec} & Section \ref{sec:primrec} & 1338 \\
    \texttt{partrec} & Section \ref{sec:partrec} & 730 \\
    \texttt{partrec\_code} & Section \ref{sec:universal} & 918 \\
    \texttt{halting} & Section \ref{sec:apps} & 354
\end{tabular}
\caption{Line counts (unadjusted) for the files in this formalization. Note that \texttt{primrec.lean} contains mostly endpoint theorems intended for presenting users with a convenient API for primitive recursion proofs.}
\label{fig:my_label}
\end{figure}
\section{Related Works}
While this is the first formalization of computability theory in Lean, there are a variety of related formalizations in other theorem provers.
\begin{itemize}
\item Zammit (1997) \cite{zammit1997} uses $n$-ary $\mu$-recursive functions with an explicit big-step semantics. Although we believe we have reproduced all the theorems in this report and more, it should be noted that this predates all the others on this list by more than 10 years.
\item Norrish achieves a substantial amount in \cite{norrish2011}, using the $\lambda$-calculus in HOL4, up to Rice's theorem and r.e.\ sets. The primary differences involve the differing model of computation and differences from working in a classical higher order logic system rather than a dependent type theory. (Lean is primarily focused on classical logic, but it permits working in intuitionistic logic, and there was no particular reason to assume LEM except in Theorem \ref{thm:post}.)
\item Asperti and Riccoti \cite{asperti2012} have formalized the construction of a universal Turing machine in Matita, but do not go as far as the halting problem or recursively enumerable sets.
\item ``Mechanising turing machines and computability theory in Isabelle/HOL'' by Xu, Zhang and Urban \cite{xu2013} builds from Turing machines, constructs a universal Turing machine, formalizes the halting problem, and relates them to abacus machines and recursive functions. But they acknowledge up front that formalizing TMs is a ``daunting prospect,'' and their formalization is much longer (although direct comparisons are misleading at best).
\item Forster and Smolka \cite{forster2017} formalize call-by-value $\lambda$-calculus in Coq, including Post's theorem and the halting problem, but they have a much greater focus on constructive mathematics and the exploration of choice principles such as Markov's principle. As Lean is not as focused on constructive type theory, we have instead chosen to focus on getting these core results with a minimum of fuss and with the most externally useful developments, so that they can perform well as an addition to Lean's standard library.
\item In ``Typing Total Recursive Functions in Coq'' \cite{larchey2017}, Larchey-Wendling shows that all total recursive functions have function witnesses in Coq. From the point of view of our paper, at least concerning total recursive functions in the sense used in computability theory, this is a consequence of the definition - a computable function has a function witness by definition, as it is a predicate on functions. Similarly, we can evaluate a partial recursive function when it is defined because of the definition of $\Part\;\alpha=\Sigma p,p\to \alpha$. The content of the theorem is then shifted to the construction of the function $\mathsf{fix}$, which was not detailed here but reduces to $\mathsf{nat.find}:(\exists n:\N.\;P(n))\to\{n\mid P(n)\}$, which ultimately relies on the same subsingleton elimination principle used in Coq.
\item In ``Formalization of the Undecidability of the Halting Problem for a Functional Language'' by Ramos et. al. \cite{ramos2018}, the authors formalize a simplified version of PVS called PVS0 suitable for translating regular PVS definitions into PVS0 and proving termination, and they also do some computability theory in this setting, including the fixed point theorem and Rice's theorem using an explicit PVS0 program. Our approach is much more abstract and generic, more suited to the mathematical theory than concrete execution models.
\end{itemize}

From our own work and the work in these alternative formalizations, we find the following ``take-home messages'':
\begin{itemize}
\item Although the standard formulation of $\mu$-recursive functions uses $n$-ary functions, and both \cite{zammit1997} and \cite{larchey2017} use $n$-ary $\mu$-recursive functions, it turns out that it is much simpler to work with unary $\mu$-recursive functions and rely on the pairing function to get additional arguments. This simplifies the statement of composition and projection significantly and decreases the reliance on dependent types.
\item There is not a significant difference between our formulation of partial recursive functions and the lambda calculus with de Bruijn variables, although we don't get the higher-order features until fairly late in the process. (Once we have \eval\ and \code\ we can use codes as higher order functions.) But it is less obvious how to get primitive recursion in the lambda calculus, and having a direct enumeration of all sets under consideration makes it easy to get things like \textsf{option.map} as primitive recursive functions early on.
\item Building ``synthetic computability'' \cite{forster2019synthetic} into the types from the beginning makes it obvious that all computable functions are Lean-computable and all partial recursive functions can be evaluated on their domain. All the work is transferred to the single function $\mathsf{fix}$, whose definition is independent of the computability library, and a complicated induction on partial recursive functions is avoided.
\item Synthetic computability is convenient when applicable, but in the presence of a ``proper'' definition of computability, they are incompatible. It is not possible to prove that all synthetically computable functions (that is, all functions) are computable, and this statement is disprovable in classical logic, so we cannot assume it to be the case. (In fact, there is a diagonalization problem here as well; even in no-axioms Lean, we cannot take the assumption that all functions are computable as an axiom without making the axiom false.)
\end{itemize}
\section{Future Work}\label{sec:future}
\subsection{Equivalences}
The most obvious next step is to show the equivalence of other formulations of computable functions: Turing machines, $\lambda$-calculus, Minsky register machines, C... the space of options is very wide here and it is easy to get carried away. Furthermore, if one holds to the thesis that partial recursive functions are the quickest lifeline out of the Turing tarpit, then one must acknowledge that this is to jump right back in, where the hardest part of the translation is fiddling with the intricacies of the target language. We are still looking for ways to do this in a more abstract way that avoids the pain. Forster and Larchey-Wendling \cite{forster2019, larchey2019} have recently made some progress in this direction, connecting Turing machines to Minsky register machines and Diophantine equations.

\subsection{Complexity theory}

This project was in part intended to set up the foundations of complexity theory. One of the often stated reasons for choosing Turing machines over other models of computation like primitive recursion is because they have a better time model. We would argue that this is not true at fine grained notions of complexity, because there is often a linear multiplicative overhead for running across the tape compared to memory models. Moreover, in the other direction we find that, at least in the case of polynomial time complexity, there are methods such as bounded recursion on notation \cite{hofmann2000} that generalize  primitive recursion methods to the definition of polynomial time computable functions, which can be used to define ${\bf P}, {\bf NP}$, and $\bf NP$-hardness at least; we are hopeful that these methods can extend to other classes, possibly by hybridizing with other models of computation as well.


\bibliography{references}

\begin{thebibliography}{10}

\bibitem{asperti2012}
Andrea Asperti and Wilmer Ricciotti.
\newblock Formalizing turing machines.
\newblock In {\em International Workshop on Logic, Language, Information, and
  Computation}, pages 1--25. Springer, 2012.

\bibitem{bauer2006}
Andrej Bauer.
\newblock {First Steps in Synthetic Computability Theory}.
\newblock {\em Electronic Notes in Theoretical Computer Science}, 155:5--31,
  2006.

\bibitem{church1936}
Alonzo Church.
\newblock An unsolvable problem of elementary number theory.
\newblock {\em American journal of mathematics}, 58(2):345--363, 1936.

\bibitem{demoura2015}
Leonardo de~Moura, Soonho Kong, Jeremy Avigad, Floris Van~Doorn, and Jakob von
  Raumer.
\newblock The lean theorem prover (system description).
\newblock In {\em International Conference on Automated Deduction}, pages
  378--388. Springer, 2015.

\bibitem{forster2019synthetic}
Yannick Forster, Dominik Kirst, and Gert Smolka.
\newblock {On synthetic undecidability in Coq, with an application to the
  Entscheidungsproblem}.
\newblock In {\em Proceedings of the 8th ACM SIGPLAN International Conference
  on Certified Programs and Proofs}, pages 38--51. ACM, 2019.

\bibitem{forster2019}
Yannick Forster and Dominique Larchey-Wendling.
\newblock {Certified undecidability of intuitionistic linear logic via binary
  stack machines and Minsky machines}.
\newblock In {\em Proceedings of the 8th ACM SIGPLAN International Conference
  on Certified Programs and Proofs}, pages 104--117. ACM, 2019.

\bibitem{forster2017}
Yannick Forster and Gert Smolka.
\newblock {Weak Call-by-Value Lambda Calculus as a Model of Computation in
  Coq}.
\newblock In {\em ITP}, 2017.

\bibitem{godel1931}
Kurt G{\"o}del.
\newblock {\"U}ber formal unentscheidbare s{\"a}tze der principia mathematica
  und verwandter systeme i.
\newblock {\em Monatshefte f{\"u}r mathematik und physik}, 38(1):173--198,
  1931.

\bibitem{hofmann2000}
Martin Hofmann.
\newblock Programming languages capturing complexity classes.
\newblock {\em ACM SIGACT News}, 31(1):31--42, Jan 2000.
\newblock \href {http://dx.doi.org/10.1145/346048.346051}
  {\path{doi:10.1145/346048.346051}}.

\bibitem{kleene1943}
Stephen~Cole Kleene.
\newblock Recursive predicates and quantifiers.
\newblock {\em Transactions of the American Mathematical Society},
  53(1):41--73, 1943.

\bibitem{larchey2017}
Dominique Larchey-Wendling.
\newblock {Typing total recursive functions in Coq}.
\newblock In {\em International Conference on Interactive Theorem Proving},
  pages 371--388. Springer, 2017.

\bibitem{larchey2019}
Dominique Larchey-Wendling and Yannick Forster.
\newblock {Hilbert's Tenth Problem in Coq}.
\newblock In {\em 4th International Conference on Formal Structures for
  Computation and Deduction (FSCD 2019)}. Schloss Dagstuhl-Leibniz-Zentrum fuer
  Informatik, 2019.

\bibitem{norrish2011}
Michael Norrish.
\newblock Mechanised computability theory.
\newblock In {\em International Conference on Interactive Theorem Proving},
  pages 297--311. Springer, 2011.

\bibitem{perlis1982}
Alan~J Perlis.
\newblock {Special feature: Epigrams on programming}.
\newblock {\em ACM Sigplan Notices}, 17(9):7--13, 1982.

\bibitem{ramos2018}
Thiago Mendon{\c{c}}a~Ferreira Ramos, C{\'e}sar Mu{\~n}oz, Mauricio
  Ayala-Rinc{\'o}n, Mariano Moscato, Aaron Dutle, and Anthony Narkawicz.
\newblock {Formalization of the Undecidability of the Halting Problem for a
  Functional Language}.
\newblock In {\em International Workshop on Logic, Language, Information, and
  Computation}, pages 196--209. Springer, 2018.

\bibitem{turing1937}
Alan~M Turing.
\newblock {On computable numbers, with an application to the
  Entscheidungsproblem}.
\newblock {\em Proceedings of the London mathematical society}, 2(1):230--265,
  1937.

\bibitem{xu2013}
Jian Xu, Xingyuan Zhang, and Christian Urban.
\newblock {Mechanising turing machines and computability theory in
  Isabelle/HOL}.
\newblock In {\em International Conference on Interactive Theorem Proving},
  pages 147--162. Springer, 2013.

\bibitem{zammit1997}
Vincent Zammit.
\newblock {A proof of the $S^m_n$ theorem in Coq}.
\newblock Technical Report 9-97, University of Kent, 1997.

\end{thebibliography}

\end{document}